\documentclass[draftclsnofoot,onecolumn]{ieeetran}
\usepackage{amsmath,amsthm,amssymb,paralist,subfigure,cite,graphicx,algorithm}
\newtheorem{lem}{Lemma}
\newtheorem{theorem}{Theorem}
\newtheorem{defn}{Definition}

\newtheorem{prop}{Proposition}

\def\mb{\mathbf}
\def\mbb{\mathbb}
\def\mc{\mathcal}

\begin{document}
\title{Measurement partitioning and observational equivalence in state estimation}
\author{Mohammadreza Doostmohammadian$^\dagger$, \emph{Student Member, IEEE}, and Usman A. Khan$^\dagger$, \emph{Senior Member, IEEE}
\thanks{
$^\dagger$ The authors are with the Department of Electrical and Computer Engineering, Tufts University, 161 College Ave., Medford, MA 02155, {\texttt{\{mrd,khan\}@ece.tufts.edu}}. This work has been partially supported by an NSF Career Award \# CCF-1350264.}}
\maketitle

\begin{abstract}
This letter studies measurement partitioning and equivalence in state estimation based on graph-theoretic principles. We show that a set of critical measurements (required to ensure LTI state-space observability) can be further partitioned into two types:~$\alpha$ and~$\beta$. This partitioning is driven by different graphical (or algebraic) methods used to define the corresponding measurements. Subsequently, we describe observational equivalence, i.e. given an~$\alpha$ (or~$\beta$) measurement, say~$y_i$, what is the set of measurements equivalent to~$y_i$, such that only one measurement in this set is required to ensure observability? Since~$\alpha$ and~$\beta$ measurements are cast using different algebraic and graphical characteristics, their equivalence sets are also derived using different algebraic and graph-theoretic principles. We illustrate the related concepts on an appropriate system digraph.   

\end{abstract}

\section{Introduction}
The concept of observability is a determining factor in state estimation. In linear but static cases, observability defines the solvablity of the set of~$p$ measurement equations to recover an~$n$-dimensional parameter, subsequently requiring at least as many measurements as the number of unknowns,~$p\geq n$, in general. Observability in LTI dynamics is more interesting since the number,~$p$, of measurements may be less than the number,~$n$, of states. Simply, an observable system possesses enough dependencies among the states, via the system matrix, that can be exploited towards state estimation with a smaller number of measurements. There are different approaches to check for observability of LTI systems: (i) algebraic method of finding the rank of the observability Gramian~\cite{luenberger1979introduction, bay}; (ii) the Popov-Belevitch-Hautus (PBH) test~\cite{hautus}; and, (iii) graph-theoretic analysis of the system digraph~\cite{pichai1984graph, rein_book, sauter:09,woude:03,liu-pnas}.

In LTI state-space observability, a significant question is to find a set of critical measurements to ensure observability. Recent literature~\cite{liu-pnas,boukhobza-recovery,asilomar11, commault-recovery} discusses different aspects and approaches towards this problem; see our prior work~\cite{icassp13} for rank-deficient systems. In these works, the LTI systems are modeled as digraphs and graph-theoretic algorithms are adapted to find the corresponding critical measurements. Since these results are structural (only depend on the non-zeros of the system matrices), they ensure what is referred to as \emph{generic observability}, i.e. the underlying LTI systems are observable for almost all choices of non-zeros in the corresponding matrices. The values for which the results do not hold lie on an algebraic variety with zero Lebesgue measure~\cite{woude:03}. 

In this letter, we first show that the set of critical measurements can be further partitioned into two types. Each of this partition is different than the other in both \emph{algebraic} and \emph{graph-theoretic} sense. We explicitly capture this algebraic and graph-theoretic characterizations. Algebraically, these partitions belong to different methods of recovering the rank of the observability Gramian. Graph-theoretically, these partitions are related to the Strongly Connected Components (SCC) and contractions in the system digraphs \cite{asilomar11,commault-recovery,icassp13,jstsp14}. We then introduce the notion of observational equivalence in state estimation. In particular, we derive a set of alternatives for each critical measurement such that if any critical measurement is not available then an equivalent measurement can be chosen to recover the loss of system observability. Clearly, since the measurements may have different algebraic and graph-theoretic characteristics, their equivalence sets are also derived using different algebraic and graph-theoretic principles.

We now describe the rest of the letter. Section~\ref{prelim} covers the preliminaries on LTI state-space descriptions and the notion of generic observability. Using this setup, we provide the problem formulation in Section~\ref{pf}. The main results on measurement partitioning are derived in Section~\ref{part}, while  observational equivalence is characterized in Section~\ref{eqv}. Finally, an illustrative example and concluding remarks are given in Sections~\ref{ex} and~\ref{con}, respectively.

\section{Preliminaries}\label{prelim}
Consider the LTI state-space dynamics as follows:
\begin{eqnarray}\label{sys1}
\mb{x}_{k+1} = A\mb{x}_k + \mb{v}_k,&\qquad& \mb{y}_k=H\mb{x}_k +\mb{r}_k,\\\label{sysCT}
\dot{\mb{x}} = A\mb{x} + \mb{v},&\qquad& \mb{y}=H\mb{x} +\mb{r},
\end{eqnarray}
where the former is the discrete-time description and the latter is the continuous-time description. Since observability in either case is identical and depends on the system matrix,~$A$, and the measurement matrix,~$H$, the treatment in this letter is applicable to both cases. Using the standard terminology,~$\mb{x}=[x_{1}~\ldots~x_{n}]^\top\in\mbb{R}^n$ is the state-space,~$\mb{y}=[y_1,\ldots,y_p]\in\mbb{R}^p$ is the measurement vector; the noise variables,~$\mb{v}$ and~$\mb{r}$, have appropriate dimensions with the standard assumptions on Gaussianity and independence. The discrete-time description is indexed by~$k\geq0$. It is well-known that the LTI descriptions above, Eqs.~\eqref{sys1} and~\eqref{sysCT}, lead to a bounded estimation error if and only if the observability Gramian, 
\begin{eqnarray}\label{pbh}
\mc{O}(A,H) = 
\left[
\begin{array}{cccc}
H^\top & (HA)^\top & \ldots & (HA^{n-1})^\top
\end{array}
\right]^\top
\end{eqnarray}
is full-rank~\cite{bay}. 

For estimation purposes, the system must be observable ($\mbox{rank}(\mc{O})=n$) with the given set of measurements. The term observability is a quantitative measure defining the ability to estimate the entire state vector,~$\mb{x}$, with bounded estimation error. Algebraic tests for observability check the rank of the Gramian,~$\mc{O}$, or the invertability of~$\mathcal{O}^\top \mathcal{O}$,~\cite{kalman:61, bay}. The PBH test~\cite{hautus}, on the other hand, is a symbolic method to test for system observability. This method checks if the matrix,
\begin{eqnarray}
\left[
\begin{array}{c}
A-sI \\
H
\end{array}
\right],
\end{eqnarray} 
is full-rank for all values of~$s \in \mathbb{C}$ where~$I$ is the~$n\times n$ identity matrix. The matrix,~$A-sI$, is full rank for all (probably complex) values of~$s$, except for the eigenvalues of~$A$. This simply implies that the PBH test has to be checked \textit{only} for these values. In other words, an LTI system is \textit{not} observable if and only if there exist a right eigenvector of~$A$ in the null space of measurement matrix,~$H$, i.e.~$\exists \mb{w} \in \mathbb{R}^{n}$ such that~$A\mb{w}=\lambda \mb{w}$ and~$H\mb{w}=0$. Both these approaches rely on the knowledge of exact values of system and measurement matrices, i.e. the numerical values of all of the elements in~$A$ and~$H$. However, it is often the case that the sparsity (zeros and non-zeros) of the system matrix is fixed while the non-zero elements change; for example, because of linearization of system parameters on different operating points \cite{Liu-nature}. For such cases, along with the computational complexity and numerical inaccuracies in computing matrix ranks, the numerical methods may not be feasible to test for observability. 

\subsection{Graph-theoretic Observability}\label{obsrv}
Instead of using the algebraic tests for observability, an alternate is a graph-theoretic approach that is described on the \emph{system digraph} as follows. Let~$\mc{X}=\{x_1,\ldots,x_n\}$ and~$\mc{Y}=\{y_1,\ldots,y_p\}$ denote the set of states and measurements, respectively. The system digraph is a \emph{directed graph} defined as~$\mc{G}_{\scriptsize \mbox{sys}} = (\mc{V}_{\scriptsize \mbox{sys}},\mc{E}_{\scriptsize \mbox{sys}})$, where~$\mc{V}_{\scriptsize \mbox{sys}}=\mc{X} \cup \mc{Y}$ is the set of nodes and~$\mc{E}_{\scriptsize \mbox{sys}}$ is the set of edges; this digraph is induced by the structure of the system and measurement matrices,~$A=\{a_{ij}\}$, and~$H=\{h_{ij}\}$. An edge,~$x_j {\rightarrow} x_i$, in~$\mc{E}_{\scriptsize \mbox{sys}}$ exists from~$x_j$ to~$x_i$ if~$a_{ij}\neq0$. Similarly, an  edge,~$x_j {\rightarrow} y_i$, in~$\mc{E}_{\scriptsize \mbox{sys}}$ exists from~$x_j$ to~$y_i$ if~$h_{ij}\neq0$. A \emph{path} from~$x_j$ to~$x_i$ (or~$y_i$) is denoted as~$x_j\overset{\scriptsize\mbox{path}}{\longrightarrow} x_i$. A path is called \emph{$\mc{Y}$-connected} (denoted by~$\overset{\scriptsize\mbox{path}}{\longrightarrow} \mc{Y}$) if it ends in a measurement. A \emph{cycle} is a path where the begin and end nodes are the same. A cycle family is a set of cycles which are mutually disjoint, i.e. they don't share any nodes. Similarly, a path and a cycle are disjoint if they do not share any node. More details on this construction can be found in~\cite{woude:03,godsil}.

Graph-theoretic (also known as generic) observability is based on structured systems theory. This method only relies on the system structure (zeros and non-zeros) and is valid for \textit{almost all} numerical values of the system parameters; the values of the system matrices for which generic observability does not hold lie on an algebraic variety of zero Lebesgue measure, see~\cite{woude:03} and references therein. We now provide the main result on generic observability--dual of the generic controllability result in~\cite{lin}.
\begin{theorem}\label{woude_thm}
A system is generically~$(A,H)$-observable if and only if in its digraph:
\begin{enumerate}[(i)]

\item Every state~$x_i$ is the begin node of a~$\mathcal{Y}$-connected path, i.e.~$x_i \overset{\scriptsize\mbox{path}}{\longrightarrow} \mc{Y}, \forall i \in \{1,\ldots,n\}$;

\item There exist a family of \textit{disjoint}~$\mathcal{Y}$-connected paths and cycles that includes all states.
\end{enumerate}
\end{theorem}
\noindent The first condition is known as \textit{accessibility} and the second as the \textit{S-rank} or \textit{matching} condition. The above conditions, however, are known to have algebraic meanings, see~\cite{shields}. We describe the algebraic connections in the following. 
\begin{prop}\label{ACprop}
Accessibility is tied with the irreducibility of the matrix~$\left[ \begin{array}{cc} A^\top & H^\top \end{array} \right]^\top$. Having an inaccessible node in the system digraph implies the existence of a permutation matrix~$P$ such that,
\begin{eqnarray}
P A P^{-1} = \left[ \begin{array}{c|c}
A_{11}& A_{12} \\\hline
0& A_{22}
\end{array} \right], ~ PH =[0 ~~| ~ H_1].
\end{eqnarray}
\end{prop}

\begin{prop}\label{SRprop}
S-rank condition is related to the structural rank of the system, i.e.
\begin{eqnarray}\label{EQ4}
\mbox{S-rank} \left[ \begin{array}{c} A\\ H \end{array} \right] = n.
\end{eqnarray}
\end{prop}
The definition of structural rank ($S$-rank in short) and its properties can be described as follows. The~$S$-rank (also called \textit{ generic rank}) is the maximal rank of a matrix,~$A$, that can be achieved by changing its non-zero elements. In the system matrix,~$A$,~$S$-rank equals the maximum number of non-zero elements in the distinct rows and columns of~$A$. In the system digraph,~$\mc{G}_{\scriptsize \mbox{sys}}$,~$S$-rank is the size of the \textit{maximum matching}, see~\cite{shields,van1999generic} for details.  

\section{Problem Formulation}\label{pf}
In this letter, we first consider the problem of finding a set of state measurements that is required for observability. We show that each such set can be partitioned into two types of measurements:~$\alpha$ and~$\beta$; these different types of measurements have different algebraic and graph-theoretic interpretations that we characterize. Algebraically, Type-$\alpha$ measurements correspond to the $S$-rank condition in Proposition~\ref{SRprop} and increase the $S$-rank of the matrix in Eq.~\eqref{EQ4}; while Type-$\beta$ measurements are tied to the accessibility condition in Proposition~\ref{ACprop}. Graph-theoretically, Type-$\alpha$ and Type-$\beta$ measurements belong to maximum matching and parent SCCs in the system digraph. We describe these results in Section~\ref{part}.

Clearly, a set of measurements that ensures observability may not be unique motivating to search for all possible sets that ensure observability. In this context, the second problem we consider is to define the states that are \textit{equivalent} in terms of observability--the equivalence relation is denoted by~`$\sim$'. In particular, if two states,~$x_i$ and~$x_j$ are observationally equivalent, i.e.~$x_i \sim x_j$, then measuring any one of them suffices for observability. Hence, the corresponding measurements are also equivalent, i.e.~$y_i\sim y_j$. We characterize this notion of observational equivalence towards state estimation in both algebraic and graph-theoretic sense. The details of this process are provided in Section~\ref{eqv}.

\section{Measurement Partitioning}\label{part}
In this section, we describe the process of measurement partitioning. Given a set of observable measurements, i.e. the matrix, $H$, such that~$(A,H)$ is observable, we partition the measurements into three types:~$\alpha$,~$\beta$, and~$\gamma$. Type-$\alpha$ and Type-$\beta$ are necessary for observability (assuming fixed~$H$) while Type-$\gamma$ measurements are redundant \cite{jstsp}. 
\begin{defn}\label{Odef}
Given system matrices,~$A$ and~$H$, a measurement is necessary for observability if and only if removing that measurement renders the system unobservable. 
\end{defn}
For a given~$H$, let~$H_{\alpha}$,~$H_{\beta}$, and~$H_\gamma$ denote the submatrices of~$H$ that represent each partition,~$\alpha$,~$\beta$, and~$\gamma$, respectively. Similarly, let~$H_{\alpha,\beta}$ denote the submatrix corresponding to both $\alpha$ and $\beta$ measurements. Using this notation the above definition can be summarized in following:
\begin{eqnarray}
\mbox{rank}\left(\mathcal{O}\left(A,\left[ 
\begin{array}{c}
H_{\alpha, \beta}\\
H_{\gamma}
\end{array}\right]
 \right)\right) = \mbox{rank}\left(\mathcal{O}\left(A,
H_{\alpha, \beta}\right)\right) = n.
\end{eqnarray}
In the rest of this section, we characterize the methods to arrive at these partitions. The results we develop are in lieu of Theorem~\ref{woude_thm} and rely on ensuring that each of the two (graph-theoretic) conditions (i) and (ii) are satisfied. It is straightforward to note that the graph-theoretic interpretation is described in Theorem~\ref{woude_thm}, while the algebraic interpretations are characterized in Propositions~\ref{ACprop} and~\ref{SRprop}.
\subsection{Graph-Theoretic}\label{MPGT}
We first consider the Type-$\alpha$ measurements. \emph{Type-$\alpha$} measurements are related to the maximum matching and contractions in the system digraph. Define the maximum matching,~$\mc{M}$, as the largest subset of edges with no common end nodes. Notice that the maximum matching is not unique. This can be best defined over the bipartite representation,~$\Gamma _A = (\mc{V}^+,\mc{V}^-,\mc{E}_{\Gamma_A})$, of system digraph, where~$\mc{V}^+=\mc{X}^+$ is the set of begin nodes and~$\mc{V}^-=\mc{X}^-\cup\mc{Y}^-$ is the set of end nodes, with the edge set defined as~$(v_i^+\in\mc{V}^+,v_j^-\in\mc{V}^-) \in \mc{E}_{\Gamma_A}$ if~$x_j {\rightarrow} x_i$ or~$x_j {\rightarrow} y_i$.
Having the maximum matching,~$\mc{M}$, let~$\delta \mc{M}^+$ represent the set of unmatched nodes defined as the nodes in~$\mc{V}^+$ not incident to the edges in~$\mc{M}$. 
\begin{defn}
A Type-$\alpha$ measurement is the measurement of an unmatched node,~$v_j\in\delta \mc{M}^+$, in the matching,~$\mc{M}$.
\end{defn}

On the other hand, \textit{Type-$\beta$} measurements are related to the Strongly-Connected Components (SCCs) in the system digraph. In a not strongly-connected digraph, define SCCs,~$\mc{S}_i$'s, as the largest strongly-connected sub-graphs. In addition, an SCC is matched, denoted by~$\mc{S}^{\circlearrowleft}_i$, if it contains a family of disjoint cycles covering all its states. A cycle is a simple example of a matched SCC. An SCC is called \textit{parent}, denoted by~$\mc{S}^{p}_i$, if it has no outgoing edges to any other SCC. Any SCC that is not parent is a \textit{child},~$\mc{S}^{c}_i$. In this regard, define \textit{partial order},~$\preceq$, as the existence of edges from one SCC to another. Mathematically, ~$\{\mc{S}_i \preceq \mc{S}_j \}$ if and only if~$v_i\overset{path}{\longrightarrow} v_j$ for some nodes~$\{v_i \in \mc{S}_i , v_j \in \mc{S}_j \}$.  Clearly we have, $\mc{S}^c \preceq \mc{S}^p$.
\begin{defn}
A Type-$\beta$ measurement is the measurement of a state in a matched parent SCC,~$\mc{S}^{\circlearrowleft p}_i$. 
\end{defn}
\subsection{Algebraic} \label{alg_part}
It is straightforward to note that Proposition~\ref{SRprop} may not be satisfied by a measurement that satisfies Proposition~\ref{ACprop}. In other words, a measurement that recovers accessibility may not improve the $S$-rank of $[A^\top~H^\top]^\top$. In this sense, we define a Type-$\alpha$ measurement as the one that improves the $S$-rank of $[A^\top~H^\top]^\top$ by~$1$. Note that when the system matrix has a full $S$-rank, there are no Type-$\alpha$ measurements because $S$-rank$(A)$ is already~$n$ and need not to be improved by the measurements. 
\begin{defn}
In the algebraic sense, the $i$th Type-$\alpha$ measurement,~$\alpha_i$, is formally defined as a measurement satisfying
\begin{eqnarray}\label{Tal_def}
S\mbox{-rank}\left(
\left[
\begin{array}{c}
A \\
 H_{\alpha_i}
\end{array}
\right]
\right) =  S\mbox{-rank}(A)+1,
\end{eqnarray}
where~$H_{\alpha_i}$ is a $1\times n$ row with a non-zero at the $\alpha_i$th location. 
\end{defn}
\noindent Each Type-$\alpha$ agent thus improves the $S$-rank condition (ii) in Theorem~\ref{woude_thm} by exactly~$1$.
\begin{defn}  
The $i$th Type-$\beta$ measurement,~$\beta_i$, does not improve the~$S$-rank, i.e.
\begin{eqnarray}\label{EQ9}
\mbox{S-rank}\left(
\left[
\begin{array}{c}
A \\
 H_{\beta_i}
\end{array}
\right]
\right) =  \mbox{S-rank}(A),
\end{eqnarray}
But, from Def.~\ref{Odef}, a Type-$\beta_i$ measurement satisfies Eq.~\eqref{EQ9} and 
\begin{eqnarray}
\mbox{rank}\left(\mathcal{O}\left(A,\left[
\begin{array}{c}
H_{\alpha} \\
H_{\beta_i}
\end{array}\right] \right)\right) = \mbox{rank}\left(\mathcal{O}\left(A,
H_{\alpha} \right)\right) + 1.
\end{eqnarray}
\end{defn}

\section{Observational Equivalence} \label{eqv}
In set theory and abstract algebra, an equivalence relation,~`$\sim$', is defined as having three properties: reflexivity, symmetry, and transitivity~\cite{abstractalgebra}. Towards observational equivalence in state estimation, reflexivity implies that every state is equivalent to itself, i.e.~$x_i \sim x_i$, symmetry implies that if~$x_i \sim x_j$ then~$x_j \sim x_i$, and transitivity implies that if~$x_i \sim x_j$ and~$x_j \sim x_m$, then~$x_i \sim x_m$. We now define observational equivalence as follows.
\begin{defn}
Let~$H_i$ denote a row with a non-zero at the~$i$th location denoting a measurement of the $i$th state. Observational equivalence among two states,~$x_i\sim x_j$, is defined as 
\begin{eqnarray*}
\mbox{rank}~\mc{O}(A, H_i) = \mbox{rank}~\mc{O}(A, H_j) = \mbox{rank}~\mc{O}\left(A, \left[
\begin{array}{c}
H_i\\
H_j
\end{array}
\right]
\right).
\end{eqnarray*}
\end{defn}
\noindent It can be easily verified that the above definition obeys the three requirements on transitivity, reflexivity, and symmetry. 

\subsection{Graph-Theoretic}
In order to define observational equivalence using graph-theoretic arguments, we need to introduce the notion of a \textit{contraction}. To this aim, we first define an auxiliary graph,~$\Gamma^\mc{M}_A$, as the graph constructed by reversing all the edges of the maximum matching,~$\mc{M}$, in the bipartite graph~$\Gamma _A$; recall Section~\ref{MPGT}. This auxiliary graph,~$\Gamma^\mc{M}_A$, is used to find the contractions in the system digraph~$\mc{G}_{\scriptsize \mbox{sys}}$ as follows. In the auxiliary graph,~$\Gamma^\mc{M}_A$, assign a contraction,~$\mc{C}_i$, to each unmatched node,~$v_j \in \delta \mc{M}^+$, as the set of all nodes in~$\delta\mc{M}^+$ reachable by alternating paths from~$v_j$. We denote $v_j(\mc{C}_i)$ as the unmatched node within the contraction,~$\mc{C}_i$. An alternating path is a path with every second edge in~$\mc{M}$. Equivalence among Type-$\alpha$ agents is thus defined using contractions and unmatched nodes within each contraction. We have the following result.
\begin{lem}
Given a contraction,~$\mc{C}_i$, the unmatched node,~$v_j(\mc{C}_i)$, within this contraction is not unique,  i.e. for $v_j(\mc{C}_i)$ and $v_g(\mc{C}_i)$, $v_j(\mc{C}_i) \in \delta \mc{M}_1^+$ and $v_g(\mc{C}_i) \in \delta \mc{M}_2^+$, where $\mc{M}_1$ and $\mc{M}_2$ are two choices of maximum matching.
\end{lem}
\begin{proof}
The proof is a direct result of non-uniqueness of maximum matching from Dulmage-Mendelson decomposition~\cite{murota}. To find a contraction, we start with a particular unmatched node, e.g. in $\delta \mc{M}_1^+$. However, once we establish the contraction, $\mc{C}_i$, within this contraction there may be multiple options for an unmatched node (in another unmatched set $\delta \mc{M}_2^+$). In other words, every contraction includes exactly one unmatched node for any choice of~$\mc{M}$. 
\end{proof}
\begin{lem}
All nodes that belong to the same contraction,~$\mc{C}_i$, are equivalent Type-$\alpha$ measurements. 
\end{lem}
\begin{proof}
The equivalence relation for the Type-$\alpha$ measurements are tied to the available choices of unmatched nodes in the corresponding contraction. Since a measurement of each unmatched node improves the $S$-rank by~$1$ and each contraction contributes to a single rank-deficiency, it is straightforward to note that the three equivalence properties are satisfied.
\end{proof}

The following establishes Type-$\beta$ equivalence.
\begin{lem} 
Two Type-$\beta$ measurements,~$\beta_i$ and $\beta_j$, of states~$x_i$ and $x_j$, are equivalent,~$\beta_i\sim \beta_j$, if they belong to the same parent SCC,~$\mc{S}^{\circlearrowleft p}_i$. An immediate corollary is that all of the states that belong to the same parent SCC are equivalent. 
\end{lem}
\begin{proof}
The proof follows form the strong connectivity of~$\mc{S}^{\circlearrowleft p}_i$, which implies accessibility of all nodes from a single node in the corresponding parent SCC,~\cite{asilomar11}.  
\end{proof}
\noindent Since the parent SCCs are disjoint components in the system digraph, equivalent sets among Type-$\beta$ measurements are disjoint. Notice that if a parent SCC is \textit{unmatched}, the measurement is of Type-$\alpha$. In such cases, a Type-$\alpha$ measurement recovers both conditions for observability in Theorem~\ref{woude_thm}. 

\subsection{Algebraic}
We now provide the algebraic interpretation of equivalence among the Type-$\alpha$ and Type-$\beta$ measurements. 
\begin{lem}\label{a_eq_lem}
Two Type-$\alpha$ measurements, $\alpha_i$ and $\alpha_j$, are equivalent, $\alpha_i\sim\alpha_j$, if and only if
{\scriptsize\begin{eqnarray} \label{a1}
\mbox{S-rank}\left(
\left[
\begin{array}{c}
 A \\
 H_{\alpha_i}
\end{array}
\right]
\right) =  
\mbox{S-rank}\left(
\left[
\begin{array}{c}
 A \\ 
 H_{\alpha_j}
\end{array}
\right]
\right)
= \mbox{S-rank}\left(
\left[
\begin{array}{c}
A \\
H_{\alpha_i} \\
H_{\alpha_j}
\end{array}
\right]
\right)
\end{eqnarray}}i.e. a collection of equivalent Type-$\alpha$ measurements improves the $S$-rank by exactly $1$.
\end{lem}
\begin{proof} 
Reflexivity and symmetry are directly induced by Eq.~\eqref{a1}. For transitivity, consider three Type-$\alpha$ measurements, $\alpha_i, \alpha_j, \alpha_m$, with~$\alpha_i\sim\alpha_j$ and~$\alpha_j\sim\alpha_m$. From Eqs.~\eqref{Tal_def} and \eqref{a1}, we have $\mbox{span}([A^\top,H_{\alpha_i}^\top]^\top)=\mbox{span}([A^\top, H_{\alpha_j}^\top]^\top)$; similarly,  $\mbox{span}([A^\top,H_{\alpha_j}^\top]^\top)=\mbox{span}([A^\top,H_{\alpha_k}^\top]^\top)$ and transitivity follows. Sufficiency also follows using similar arguments.
\end{proof}
It is noteworthy tht the notion of (row) span in Lemma~\ref{a_eq_lem} is driven by $S$-rank and is to be considered as the maximal span over all possible choices of non-zeros in the corresponding matrix. The following lemma establishes Type-$\beta$ equivalence.

\begin{lem} 
Let~$H_\alpha$ denote the Type-$\alpha$ measurement matrix. Two Type-$\beta$ measurements, $\beta_i$ and $\beta_j$, are equivalent, when
\begin{eqnarray}
\mbox{rank}\left(\mathcal{O}\left(A,\left[
\begin{array}{c}
H_{\alpha} \\
H_{\beta_i}
\end{array}\right] \right)\right) = \mbox{rank}\left(\mathcal{O}\left(A,\left[
\begin{array}{c}
H_{\alpha} \\
H_{\beta_j}
\end{array}\right] \right)\right)  \nonumber \\
= \mbox{rank}\left(\mathcal{O}\left(A,\left[
\begin{array}{c}
H_{\alpha} \\
H_{\beta_i}\\
H_{\beta_j}
\end{array}\right] \right)\right) =
\mbox{rank}\left(\mathcal{O}\left(A,
H_{\alpha} \right)\right) + 1.
\end{eqnarray}
\end{lem}
\begin{proof} Reflexivity and symmetry are trivial. Transitivity follows from the fact that equivalent Type-$\beta$ measurements belong to the \textit{same irreducible block} of $A$ (see \cite{murota}).   
\end{proof} 

\section{Illustration} \label{ex}
Consider an LTI dynamical system with~$20$ state variables. The system digraph associated with the structure of the system matrix is given in Fig~\ref{fig_graph}.  This digraph has three contractions,~\{\{2,7,9\},\{4,15\},\{10,12\}\}, constituting the equivalent~Type-$\alpha$ sets; and two matched parent SCCs,~$\{\{11,12,13,14\},\{9\}\}$, constituting the equivalent~Type-$\beta$ sets (the SCC,~$\{16,17,18\}$, e.g., has an outgoing edge and hence is not parent). Three unmatched nodes each from a contraction make the~Type-$\alpha$ sets:~$\alpha_1 \in \{2,7,9\}, \alpha_2 \in \{10,12\}, \alpha_3 \in \{4,15\}$. Notice that both~Type-$\beta$ sets share nodes with the Type-$\alpha$ sets. Therefore, at least three measurements, e.g.~$\{4,9,12\}$, are necessary. In the case of not observing a shared~$\alpha$/$\beta$ state, e.g.~$12$, more than three observations are required; for example,~$\{4,9,10,13\}$ is another set of necessary measurements.   
\begin{figure}[!tbph]
\centering
\subfigure{\includegraphics[height=1.3in]{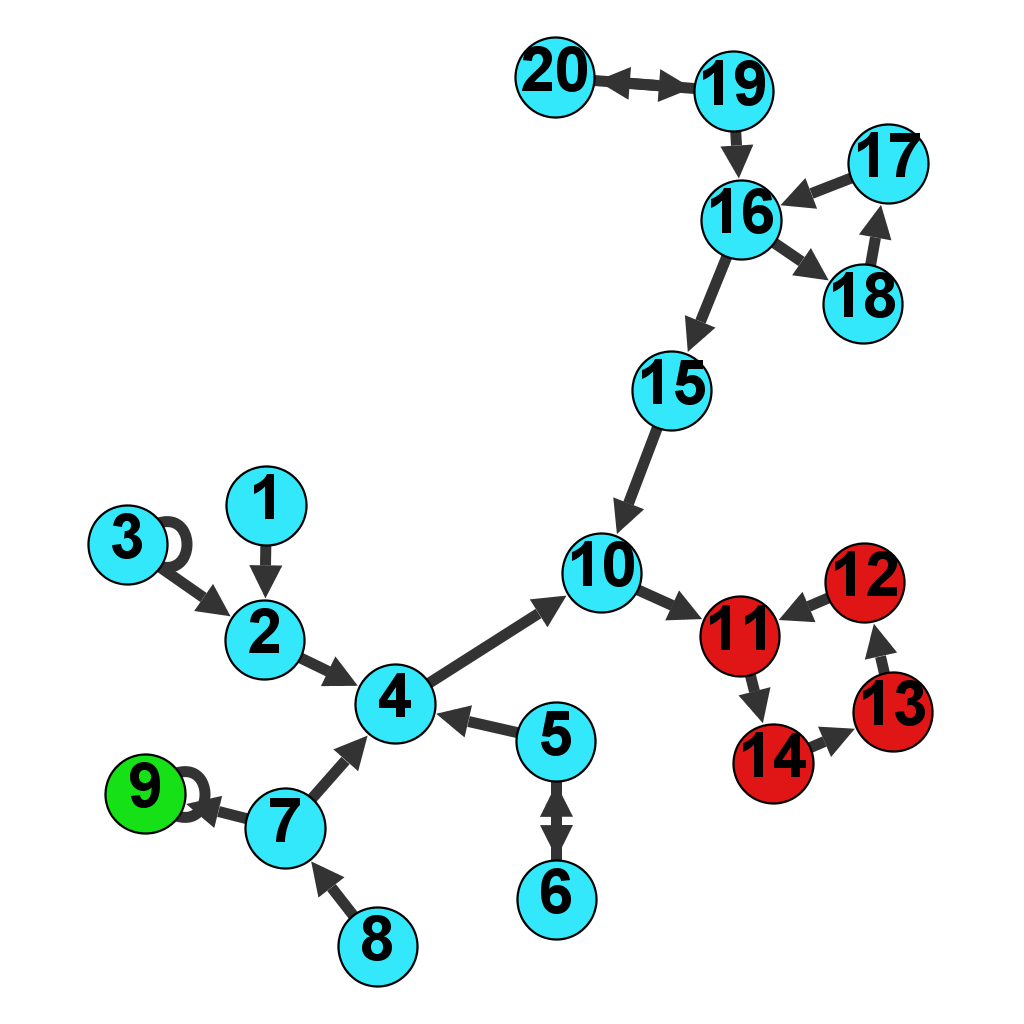}}
\subfigure{\includegraphics[height=1.3in]{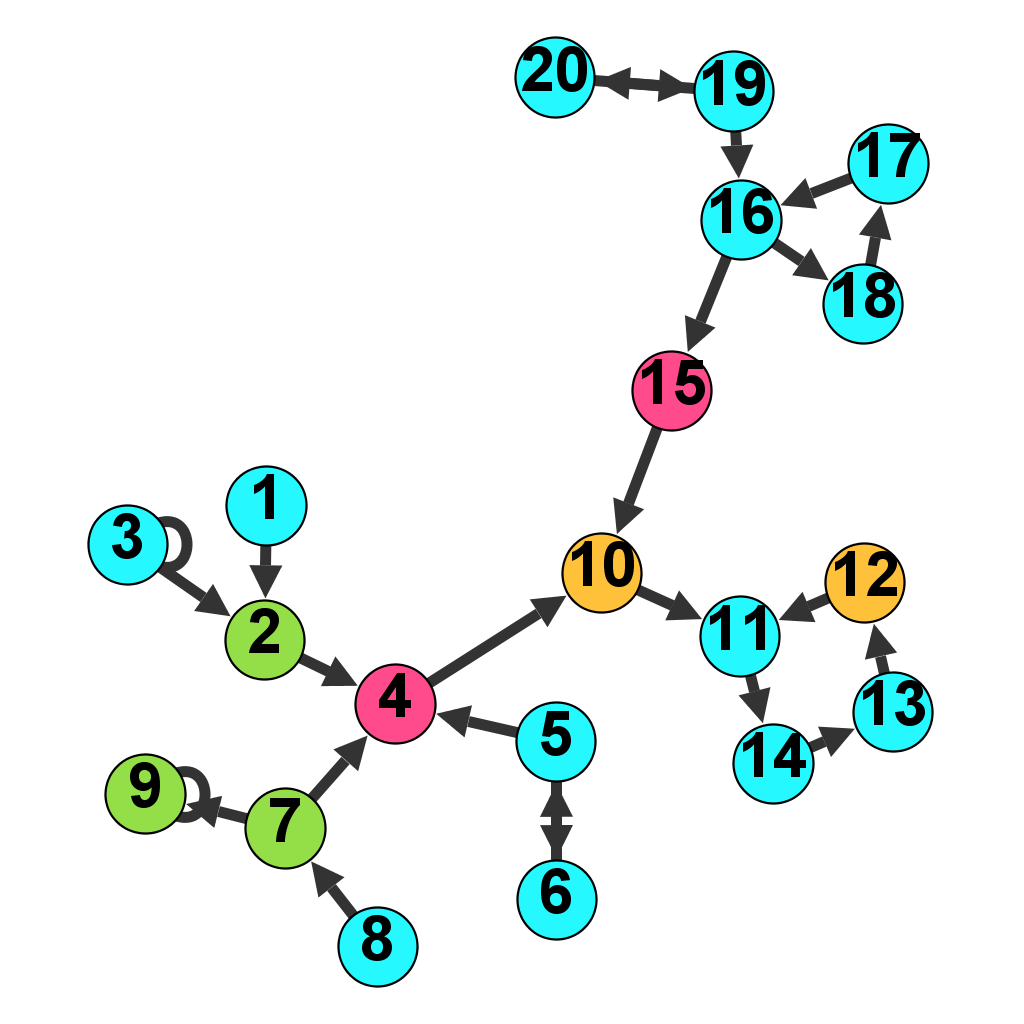}}
\caption{(Left) Type~$\beta$ equivalence sets are shown in red green. (Right) Type-~$\alpha$ equivalent sets are shown in orange, purple, and green.}
\label{fig_graph}
\end{figure}

\section{Conclusion} \label{con}
In this letter, we characterize measurement partitioning and observational equivalence in state estimation. We first derive both graph-theoretic and algebraic representations of two different classes of critical measurements, Type-$\alpha$ and Type-$\beta$, required for observability. This twofold construction of partitions leads to establishing the notion of equivalence among both Type-$\alpha$ and Type-$\beta$ measurements with different graph-theoretic and algebraic interpretations.
Notice that, there are combinatorial algorithms in \textit{polynomial order} to find partial order of SCCs,\cite{tarjan,algorithm}, maximum matching, and contractions in the system digraphs\cite{murota}.

\newpage
\bibliographystyle{IEEEbib}
\bibliography{bibliography}
\end{document}